\documentclass{amsart}

\usepackage{amsfonts}
\usepackage{amsmath}
\usepackage{amssymb}
\usepackage{graphicx}

\newtheorem{theorem}{Theorem}
\theoremstyle{plain}

\newtheorem{case}{Case}

\newtheorem{corollary}{Corollary}

\newtheorem{definition}{Definition}
\newtheorem{example}{Example}

\newtheorem{lemma}{Lemma}

\newtheorem{remark}{Remark}

\numberwithin{equation}{section}

\begin{document}
\title[Non-archimedean generalized Bessel potentials]{Non-archimedean generalized Bessel potentials and their applications} \subjclass{}
\thanks{* Anselmo Torresblanca-Badillo\\
atorresblanca@uninorte.edu.co} 

\begin{abstract}
This article describes a class of pseudo-differential operators
\begin{equation*}
(\mathcal{A}^{\alpha}\varphi)(x)=\mathcal{F}^{-1}_{\xi \rightarrow
x}\left(\left[\max\{|\boldsymbol{\psi}_{1}(||\xi||_{p})|,|\boldsymbol{\psi}_{2}(||\xi||_{p})|\}\right]^{-\alpha}\widehat{\varphi}(\xi)\right),
\end{equation*}
$\varphi\in \mathcal{D}(\mathbb{Q}_{p}^{n})$ and
$\alpha\in\mathbb{C}$; here
$\left[\max\{|\boldsymbol{\psi}_{1}(||\xi||_{p})|,|\boldsymbol{\psi}_{2}(||\xi||_{p})|\}\right]^{-\alpha}$
is the symbol of the operator $\mathcal{A}^{\alpha}$. These
operators can be seen as a generalization of the Bessel potentials
in the $p$-adic context. We show that the family
$\left(K_{\alpha}\right)_{\alpha>0}$ of convolution kernels attached
to generalized Bessel potentials $\mathcal{A}^{\alpha}$, $\alpha>0$,
determine a convolution semigroup on $\mathbb{Q}_{p}^{n}$. Imposing
certain conditions we have that $K_{\alpha}$, $\alpha>0$, is a
probability measure on $\mathbb{Q}_{p}^{n}$. Moreover, we will study
certain properties corresponding to the Green function of the
operator $\mathcal{A}^{\alpha}$ and we show that heat equations,
naturally associated to these operators, describes the cooling (or
loss of heat) in a given region over time.
\end{abstract}

\author{Anselmo Torresblanca-Badillo}
\email{atorresblanca@uninorte.edu.co}
\address{Universidad del Norte, Departamento de Matem\'aticas y Estad\'istica, Km. 5 V\'ia Puerto Colombia. Barranquilla, Colombia.}
\maketitle

\section{\protect\bigskip Introduction}

In the latest years, there has been a strong interest on the
non-archimedean pseudo-differential operators due its connections
with $p$-adic pseudo-differential equations that describe certain
physical models, see e.g. \cite{Aguilar-Cruz-Estala-2020},
\cite{Albeverio et al}, \cite{Antoniouk-Khrennikov-Kochubei},
\cite{C-Ch-G-2020}, \cite{D-K-K-V}, \cite{K-K-M-2016},
\cite{Kochubei-2001}, \cite{O-K-2017}, \cite{P-K-S-O-C-2019},
\cite{V-V-Z}, \cite{Zu-lib1}, and the references therein.
Particularly, pseudo-differential operators whose symbols are
associated with negative definite functions on the $p$-adic numbers, see \cite{Gu-To-1}, \cite{Gu-To-2}, \cite{To-Z-2}, \cite{To-Z}.\\
In this article, we introduce a new class of non-archimedean
pseudo-differential operators (called the generalized Bessel
potentials) associated with negative definite functions in the
$p$-adic context and in arbitrary dimension. These operators have
the form
\begin{equation*}
(\mathcal{A}^{\alpha}\varphi)(x)=\mathcal{F}^{-1}_{\xi \rightarrow
x}\left(\left[\max\{|\boldsymbol{\psi}_{1}(||\xi||_{p})|,|\boldsymbol{\psi}_{2}(||\xi||_{p})|\}\right]^{-\alpha}\widehat{\varphi}(\xi)\right)
, \ \ \varphi\in \mathcal{D}(\mathbb{Q}_{p}^{n}),
\end{equation*}
where $\alpha\in\mathbb{C}$, $\mathcal{F}^{-1}_{\xi \rightarrow x}$
denotes the inverse Fourier transform, $\widehat{\varphi}$ is the
Fourier transform of $\varphi$, $\mathcal{D}(\mathbb{Q}_{p}^{n})$
denotes the $\mathbb{C}-$ vector space of Bruhat-Schwartz functions
over $\mathbb{Q}_{p}^{n}$ and the function
$\left[\max\{|\boldsymbol{\psi}_{1}(||\xi||_{p})|,|\boldsymbol{\psi}_{2}(||\xi||_{p})|\}\right]^{-\alpha}$
is called symbol of the operator $\mathcal{A}^{\alpha}$. The
functions $\boldsymbol{\psi}_{1},
\boldsymbol{\psi}_{2}:\mathbb{Q}_{p}^{n}\rightarrow \mathbb{C}$ are
continuous, negative definite and radials.\\
Taking, $\boldsymbol{\psi}_{1}=1$ and
$\boldsymbol{\psi}_{2}=||\cdot||_{p}$, then by \cite[Proposition
7.4-$(iii)$]{Berg-Gunnar} and \cite[Example 3.4]{To-Z-2} we have
that $\mathcal{A}^{\alpha}$ correspond to the Bessel potential
studied at \cite{Gu-To-2} and \cite{Taibleson}.\\
We are interested in the pseudo-differential operators
$\mathcal{A}^{\alpha}$ due to our interest for contemporary physical
theories, in particular, study new Cauchy problems (or $p$-adic heat
equations), which governs the temperature distribution in an object
over time. The corresponding equations takes the form
\begin{equation*}
\left\{
\begin{array}{ll}
\frac{\partial u}{\partial t}(x,t)=\mathcal{A}^{\alpha}u(x,t), \ \  t\in[0,\infty), \ \ x\in\mathbb{Q}_{p}^{n}  \\
&  \\
u(x,0)=u_{0}(x)\in \mathcal{D}(\mathbb{Q}_{p}^{n})\text{. } &
\end{array}
\right.
\end{equation*}
It is worth noting that is that unlike the fundamental solutions
studied at \cite{Antoniouk-Khrennikov-Kochubei}, \cite{Gu-To-1},
\cite{Khrennikov-Kochubei}, \cite{Kochubei-2001}, \cite{To-Z-2},
\cite{To-Z}, \cite{Zu-lib1}, et al., in our case, the fundamental
solutions $Z(x,t)$ of the new Cauchy problems satisfies $Z(x,t)\leq
0$ for all $x\in \mathbb{Q}_{p}^{n}\backslash \left\{0\right\}$ and
$t>0$, see Theorem \ref{theorem_Z}. Therefore, we have observed that
the new above $p$-adic heat equations, give temperatures less than
absolute zero, which means from the physical point of view,  that
this equations describe the loss of heat over time.\\
There are many physical and mathematical motivations to employ
$p$-adic analysis in investigation of mathematical and theoretical
aspects of modern quantum physics, such as the properties of the
Green's function over the field of $p$-adic numbers. Green's
functions are a necessary stage in the formulation of the new
$p$-adic quantum theory. For further details the reader may consult
\cite{Bikulov-1991}, \cite{Khre-Kozy-Zu-libro}, \cite{V-V-1989}, and
the references therein.\\
Motivated by the above, in this article we introduce a new class of
$n$-dimensional $p$-adic Green function $G$ attached to
pseudo-differential operator $\mathcal{A}^{\alpha}$ and given by
\begin{equation*}
G(x):=\mathcal{F}^{-1}_{\xi \rightarrow
x}\left(\frac{1}{m^{2}+\left[\max\{|\boldsymbol{\psi}_{1}(||\xi||_{p})|,|\boldsymbol{\psi}_{2}(||\xi||_{p})|\}\right]^{-\alpha}
}\right), \text{  } m>0.
\end{equation*}
On the other hand, since the analysis of convolution semigroups of
probability measures on locally compact abelian groups play an
important role in probability theory,  functional analysis and
potential theory, see e.g. \cite{Berg-Gunnar}. In this article, we
are interested in obtaining a new family of convolution semigroups
of probability measures on $\mathbb{Q}_{p}^{n}$.\\
The article is organized as follows: In Section \ref{Fourier
Analysis}, we will collect some basic results on the $p$-adic
analysis and fix the notation that we will use through the article.
In Section \ref{Preliminary Results}, we will introduce a new class
of non-archimedean pseudo-differential operators on the space
$\mathcal{D}(\mathbb{Q}_{p}^{n})$ which we call generalized Bessel
potentials. In Section \ref{Section_convolution_kernel}, we study
certain properties corresponding to the convolution kernels
$K_{\alpha}$, attached to generalized Bessel potentials, see Theorem
\ref{Theorem_Convolution_Kernel} and Corollary \ref{Corollary_1}.
Moreover, we show that the family
$\left(K_{\alpha}\right)_{\alpha>0}$ determine a convolution
semigroup on $\mathbb{Q}_{p}^{n}$, see Theorem
\ref{convolution_semigroup}. Imposing certain conditions we have
that $K_{\alpha}$, $\alpha>0$, is a probability measure on
$\mathbb{Q}_{p}^{n}$, see Corollary
\ref{Corollary_probability_measure}. In Section \ref{Green and
Heat_Kernel}, we will study certain properties corresponding to the
Green function and the heat Kernel attached to operator
$\mathcal{A}^{\alpha}$, see Theorem \ref{properties_Green} and
Theorem \ref{theorem_Z}, respectively.

\section{\label{Fourier Analysis} Fourier Analysis on $\mathbb{Q}_{p}^{n}$: Essential Ideas}

\subsection{The field of $p$-adic numbers}

Along this article $p$ will denote a prime number. The field of
$p-$adic numbers $\mathbb{Q}_{p}$ is defined as the completion of
the field of rational numbers $\mathbb{Q}$ with respect to the
$p-$adic norm $|\cdot |_{p}$, which is defined as
\begin{equation*}
\left\vert x\right\vert _{p}=\left\{
\begin{array}{lll}
0\text{,} & \text{if} & x=0 \\
&  &  \\
p^{-\gamma }\text{,} & \text{if} & x=p^{\gamma }\frac{a}{b}\text{,}
\end{array}
\right.
\end{equation*}
where $a$ and $b$ are integers coprime with $p$. The integer $\gamma
:=ord(x) $, with $ord(0):=+\infty $, is called the\textit{\
}$p-$\textit{adic order of} $x$.\\
Any $p-$adic number $x\neq 0$ has a unique expansion of the form
\begin{equation*}
x=p^{ord(x)}\sum_{j=0}^{\infty }x_{j}p^{j},
\end{equation*}%
where $x_{j}\in \{0,1,2,\dots ,p-1\}$ and $x_{0}\neq 0$. By using this
expansion, we define \textit{the fractional part of }$x\in \mathbb{Q}_{p}$,
denoted $\{x\}_{p}$, as the rational number
\begin{equation*}
\left\{ x\right\} _{p}=\left\{
\begin{array}{lll}
0\text{,} & \text{if} & x=0\text{ or }ord(x)\geq 0 \\
&  &  \\
p^{ord(x)}\sum_{j=0}^{-ord_{p}(x)-1}x_{j}p^{j}\text{,} & \text{if} &
ord(x)<0.
\end{array}
\right.
\end{equation*}
We extend the $p$-adic norm to $\mathbb{Q}_{p}^{n}$ by taking
\begin{equation*}
||x||_{p}:=\max_{1\leq i\leq n}|x_{i}|_{p},\text{ for
}x=(x_{1},\dots ,x_{n})\in \mathbb{Q}_{p}^{n}.
\end{equation*}
For $r\in \mathbb{Z}$, denote by $B_{r}^{n}(a)=\{x\in
\mathbb{Q}_{p}^{n};||x-a||_{p}\leq p^{r}\}$ \textit{the ball of
radius }$p^{r}$ \textit{with center at} $a=(a_{1},\dots ,a_{n})\in
\mathbb{Q}_p^n$, and take $B_{r}^{n}(0)=:B_{r}^{n}$. Note that $
B_{r}^{n}(a)=B_{r}(a_{1})\times \cdots \times B_{r}(a_{n})$, where $
B_{r}(a_{i}):=\{x\in \mathbb{Q}_{p};|x_{i}-a_{i}|_{p}\leq p^{r}\}$
is the one-dimensional ball of radius $p^{r}$ with center at
$a_{i}\in \mathbb{Q}_{p}$. The ball $B_{0}^{n}$ equals the product
of $n$ copies of $B_{0}=\mathbb{Z}_{p}$, \textit{the ring of
}$p-$\textit{adic integers of }$\mathbb{Q}_{p}$. We also denote by
$S_{r}^{n}(a)=\{x\in \mathbb{Q}_p^n;||x-a||_{p}=p^{r}\}$ \textit{the
sphere of radius }$p^{r}$ \textit{with center at} $a=(a_{1},\dots
,a_{n})\in \mathbb{Q}_p^n$, and take $S_{r}^{n}(0)=:S_{r}^{n}$. The
balls and spheres are both open and
closed subsets in $\mathbb{Q}_p^n$. \\
As a topological space $\left(\mathbb{Q}_p^n,||\cdot ||_{p}\right) $
is totally disconnected, i.e. the only connected subsets of
$\mathbb{Q}_p^n$ are the empty set and the points. A subset of
$\mathbb{Q}_p^n$ is compact if and only if it is closed and bounded
in $\mathbb{Q}_p^n$, see e.g. \cite[Section 1.3]{V-V-Z}, or
\cite[Section 1.8]{Albeverio et al}. The balls and spheres are
compact subsets. Thus $\left(\mathbb{Q} _p^n,||\cdot
||_{p}\right) $ is a locally compact topological space.\\
We will use $\Omega \left( p^{-r}||x-a||_{p}\right) $ to denote the
characteristic function of the ball $B_{r}^{n}(a)$. We will use the
notation $1_{A}$ for the characteristic function of a set $A\subset
\mathbb{Q}_{p}^{n}$. Along the article $ d^{n}x$ will denote a Haar
measure on $\mathbb{Q}_{p}^{n}$ normalized such that
$\int_{\mathbb{Z}_{p}^{n}}d^{n}x=1.$

\subsection{Some function spaces}

A complex-valued function $f$ defined on $\mathbb{Q}_{p}^{n}$ is
called \textit{locally constant} if for any $x\in\mathbb{Q}_{p}^{n}$
there exist an integer $l:=l(x)$ such that
\begin{equation*}
f(x)=f(x^{\prime})\text{ for all } x^{\prime }\in B_{l}^{n}(x).
\end{equation*}
Equivalently, there exists a clopen partition $U $ of
$\mathbb{Q}_{p}^{n}$ such that $f$ is constant on each element of $
U$.\\
Denote by {\LARGE$\varepsilon$}$(\mathbb{Q}_{p}^{n})$ the linear
space of locally constant $\mathbb{C}$-value functions on
$\mathbb{Q}_{p}^{n}$.\\
A function $\varphi :\mathbb{Q}_p^n\rightarrow \mathbb{C}$ is called
a \textit{Bruhat-Schwartz function (or a test function)} if it is
locally constant with compact support.\\
The $ \mathbb{C}$-vector space of Bruhat-Schwartz functions is
denoted by $\mathcal{D}(\mathbb{Q}_{p}^{n})=:\mathcal{D}$. Let
$\mathcal{D}^{\prime }(\mathbb{Q}_{p}^{n})=:\mathcal{D}^{\prime }$
denote the set of all continuous functional (distributions) on
$\mathcal{D}$. The natural pairing $\mathcal{D}^{\prime
}(\mathbb{Q}_{p}^{n})\times
\mathcal{D}(\mathbb{Q}_{p}^{n})\rightarrow \mathbb{C}$ is denoted as
$\left<T,\varphi \right>$ for $T\in \mathcal{D}^{\prime
}(\mathbb{Q}_p^n)$ and $\varphi \in \mathcal{D}(\mathbb{Q}_p^n)$,
see e.g. \cite[Section
4.4]{Albeverio et al}.\\
Denote by $L_{loc}^{1}(\mathbb{Q}_{p}^{n})$ the set of functions
$f:\mathbb{Q}_p^{n}\rightarrow \mathbb{C}$ such that $f\in L^{1}(K)$
for any compact $K\subset \mathbb{Q}_p^{n}$. Every $f\in $
$L_{loc}^{1}(\mathbb{Q}_p^{n})$ defines a distribution $f\in
\mathcal{D}^{\prime }\left(\mathbb{Q}_p^n\right)$ by the formula
\begin{equation*}
\left<f,\varphi \right> =\int_{\mathbb{Q}_p^n}f\left( x\right)
\varphi \left(x\right>d^{n}x.
\end{equation*}
Such distributions are called \textit{regular distributions}.\\
Given $\rho \in \lbrack 0,\infty )$, we denote by $L^{\rho
}\left(\mathbb{Q}_p^n,d^{n}x\right) =L^{\rho
}\left(\mathbb{Q}_p^n\right) :=L^{\rho },$ the
$\mathbb{C}-$vector space of all the complex valued functions $g$
satisfying $\int_{\mathbb{Q}_p^n}\left\vert g\left( x\right)
\right\vert ^{\rho }d^{n}x<\infty $, $ L^{\infty }\allowbreak
:=L^{\infty }\left(\mathbb{Q}_p^n\right) =L^{\infty
}\left(\mathbb{Q}_p^n,d^{n}x\right) $ denotes the $\mathbb{C}
-$vector space of all the complex valued functions $g$ such that the
essential supremum of $|g|$ is bounded.\\
Let denote by $C(\mathbb{Q}_p^n,\mathbb{C})=:C_{\mathbb{C}}$ the
$\mathbb{C}-$vector space of all the complex valued functions which
are continuous, by
$C(\mathbb{Q}_{p}^{n},\mathbb{R})=:C_{\mathbb{R}}$ the
$\mathbb{R}-$vector space of continuous functions. Set
\begin{equation*}
C_{0}(\mathbb{Q}_p^n,\mathbb{C}):=C_{0}(\mathbb{Q}_p^n)=\left\{
f:\mathbb{Q}_p^n\rightarrow \mathbb{C} ;\text{ }f\text{ is
continuous and }\lim_{||x||_{p}\rightarrow \infty }f(x)=0\right\} ,
\end{equation*}
where $\lim_{||x||_{p}\rightarrow \infty }f(x)=0$ means that for
every $\epsilon >0$ there exists a compact subset $B(\epsilon )$
such that $|f(x)|<\epsilon $ for $x\in \mathbb{Q}_p^n\backslash
B(\epsilon ).$ We recall that
$(C_{0}(\mathbb{Q}_p^n,\mathbb{C}),||\cdot ||_{L^{\infty }})$ is
a Banach space.

\subsection{Fourier transform}

Set $\chi_{p}(y)=\exp (2\pi i\{y\}_{p})$ for $y\in \mathbb{Q}_{p}$.
The map $\chi _{p}(\cdot )$ is an additive character on
$\mathbb{Q}_{p}$, i.e. a continuous map from
$\left(\mathbb{Q}_{p},+\right) $ into $S$ (the unit circle
considered as multiplicative group) satisfying $\chi
_{p}(x_{0}+x_{1})=\chi _{p}(x_{0})\chi _{p}(x_{1})$,
$x_{0},x_{1}\in\mathbb{Q}_{p}$. The additive characters of
$\mathbb{Q}_{p}$ form an Abelian group which is isomorphic to
$\left(\mathbb{Q}_{p},+\right) $, the isomorphism is given by $\xi
\mapsto \chi _{p}(\xi x)$, see e.g. \cite[Section 2.3]{Albeverio et
al}. \\
Given $x=(x_{1},\dots ,x_{n}),$ $\xi =(\xi _{1},\dots ,\xi
_{n})\in \mathbb{Q}_p^n$, we set $x\cdot \xi
:=\sum_{j=1}^{n}x_{j}\xi _{j}$. If $f\in L^{1}(\mathbb{Q}_p^n)$,
its Fourier transform is defined by
\begin{equation*}
(\mathcal{F} f)(\xi)=\mathcal{F}_{x\rightarrow \xi
}(f)=\widehat{f}(\xi):=\int_{\mathbb{Q}_p^n}\chi_{p}(\xi \cdot
x)f(x)d^{n}x,\quad \text{for }\xi \in \mathbb{Q}_p^n.
\end{equation*}
The inverse Fourier transform of a function $f\in
L^{1}(\mathbb{Q}_p^n)$ is
\begin{equation*}
(\mathcal{F}^{-1}f)(x)=\mathcal{F}^{-1}_{\xi\rightarrow x
}(f)=\int_{\mathbb{Q}_p^n}\chi _{p}(-x \cdot
\xi)f(\xi)d^{n}\xi,\quad \text{for }x \in \mathbb{Q}_p^n.
\end{equation*}
The Fourier transform is a linear isomorphism from
$\mathcal{D}(\mathbb{Q}_{p}^{n})$ onto itself satisfying
\begin{equation*}
(\mathcal{F}(\mathcal{F}f))(\xi )=f(-\xi ),  \label{FF(f)}
\end{equation*}
for every $f\in \mathcal{D}(\mathbb{Q}_{p}^{n})$, see e.g.
\cite[Section 4.8]{Albeverio et al}.\\
The Fourier transform $\mathcal{F}(f)=\mathcal{F}_{x\rightarrow \xi
}(f)=\widehat{f}$ of a distribution $f$ is defined by the relation
\begin{equation*}
\left<\mathcal{F}(f),\varphi\right>=\left<f,\mathcal{F}(\varphi)\right>,
\text{ for all } \varphi\in \mathcal{D}(\mathbb{Q}_{p}^{n}).
\end{equation*}
The Fourier transform $f\rightarrow \mathcal{F}(f)$ is a linear
isomorphism from $\mathcal{D}'(\mathbb{Q}_{p}^{n})$ onto
$\mathcal{D}'(\mathbb{Q}_{p}^{n})$, see e.g. \cite[Section
4.9]{Albeverio et al}.

\section{\label{Preliminary Results} Generalized
Bessel potentials}

The goal of this section is to introduce a large class of
non-archimedean pseudo-differential operators on
$\mathcal{D}(\mathbb{Q}_{p}^{n})$ which we call generalized Bessel
potentials.

\begin{definition} \label{neg_def}
A function $\psi:\mathbb{\mathbb{Q}}_{p}^{n}\rightarrow \mathbb{C}$
is called negative definite, if
\begin{equation*}
\sum\nolimits_{i,j=1}^{m}\left(\psi(x_{i})+\overline{\psi(x_{j})}
-\psi(x_{i}-x_{j})\right) \lambda _{i}\overline{\lambda _{j}}\geq 0
\label{negative definite}
\end{equation*}
for all $m\in \mathbb{N}\backslash \{0\},$ $x_{1},\ldots ,x_{m}\in $
$ \mathbb{\mathbb{Q}}_{p}^{n},$ $\lambda _{1},\ldots ,\lambda _{m}$
$\in $ $\mathbb{C}$.
\end{definition}

\begin{example}\label{example neg def}
\begin{enumerate}
\item[(i)] \cite[Example 3.5]{To-Z-2} For $\alpha$ and $\beta$ positive real numbers, we have that $\boldsymbol{\psi}(||\xi||_{p})=\alpha \|\xi\|_{p}^{\beta}$ is a negative
definite function, where $\|\cdot\|_{p}$ denotes the $p$-adic norm
on $\mathbb{Q}_{p}^{n}$. Note that $\boldsymbol{\psi}$ is increasing
function with respect to $\|\cdot\|_{p}$.
\item[(ii)] Let $f(\xi)\in \mathcal{\mathbf{\mathbb{Q}}}_{p}^{n}[\xi _{1},\ldots
,\xi _{n}]$ be a non-constant polynomial. We say that $f(\xi )$ is
an elliptic polynomial of degree $d$, if it satisfies: $(a)$
$f(\xi)$ is a homogeneous polynomial of degree $d$, and $(b)$
$f(\xi)=0\Leftrightarrow \xi =0$. For any $\beta>0$ we have that
$|f(\xi)|_{p}^{\beta}$ is a negative definite function, where
$|\cdot|_{p}$ denotes the $p$-adic norm on $\mathbb{Q}_{p}$, see
\cite[Theorem 3]{Gu-To-1}. On the other hand, by \cite[Lemma
25]{Zu-lib1}, there exist positive constants $C_{0}=C_{0}(f),$ $
C_{1}=C_{1}(f)$ such that
\begin{equation*}
C_{0}||\xi ||_{p}^{d}\leq |f(\xi )|_{p}\leq C_{1}||\xi
||_{p}^{d},\text{ for every }\xi \in
\mathcal{\mathbf{\mathbb{Q}}}_{p}^{n}.
\end{equation*}
\item[(iii)] \cite[Lemma 3.8]{To-Z-2} Set $\psi
_{0}\left(  \xi\right)
:=\sum_{j=1}^{\infty}c_{j}||\xi||_{p}^{\alpha_{j}}$ with
$c_{j}\geq0$, $\alpha_{j}\in\mathbb{N}$ such that the real series
$\sum_{j=1}^{\infty}c_{j}y^{\alpha_{j}}$ defines a non-constant real
function. Then for any $j\in \mathbb{N}\backslash\{0\}$,
\begin{equation*}
\boldsymbol{\psi}(||\xi||_{p}):=e^{e^{.^{.^{.e^{\psi_{0}\left(
\xi\right)  }}}}}\text{, }j-\text{powers }
\end{equation*}
is a continuous and negative definite function on
$\mathbb{Q}_{p}^{n}$. In this case there is a fixed positive
constant $\beta:=\beta(\boldsymbol{\psi})$ such that for any $\xi
\in \mathbb{Q}_{p}^{n}$ we have that
$\boldsymbol{\psi}(||\xi||_{p})>||\xi||_{p}^{\beta}$. Moreover, note
that $\boldsymbol{\psi}$ is increasing function with respect to
$\|\cdot\|_{p}$.
\item[(iv)] \cite[Remark 3-$(ii)$]{To-Z} We set $\mathbb{R}_{+}:=\{x\in\mathbb{R}:x\geq0\}$.
Let $J:$\textbf{\ }$\mathbb{Q}_{p}^{n}\rightarrow\mathbb{R}_{+}$ be
a radial (i.e. $J(x)=J(||x||_{p})$) and continuous function such
that $\int_{\mathbb{Q}_{p}^{n}}J(||x||_{p})d^{n}x=1$. Then, the
function $\widehat{J}(0)-\widehat{J}(\left\Vert \xi \right\Vert
_{p})=1-\widehat{J}(\left\Vert \xi \right\Vert _{p})$ is negative
definite. Moreover, $0\leq 1-\widehat{J}(\left\Vert \xi \right\Vert
_{p})\leq 2$ for all $\xi \in \mathbb{Q}_{p}^{n}$.
\item[(v)] \cite[Proposition 7.4-($iii$)]{Berg-Gunnar} The
non-negative constant functions are negative definite.
\end{enumerate}
\end{example}

\begin{lemma} \label{radial-locally constant}
Let $f:\mathbb{Q}_{p}^{n}\rightarrow \mathbb{C}$ be a radial
function, i.e. $f(x)=f(||x||_{p})$ for all $x\in
\mathbb{Q}_{p}^{n}$. Then $f$ is a locally constant function.
\end{lemma}

\begin{proof}
By \cite[Sections 1.8 and 1.10]{Albeverio et al} we have that
\begin{equation*}
\mathbb{Q}_{p}^{n}=\bigsqcup_{\gamma \in \mathbb{Z}}S_{\gamma}^{n}.
\end{equation*}
Let $x\in \mathbb{Q}_{p}^{n}$ fixed. Then, there exists a unique
$\gamma:=\gamma(x) \in \mathbb{Z}$ such
that $x\in S_{\gamma}^{n}$.\\
Since $f$ is a constant function on $S_{\gamma}^{n}$, then by
\cite[Proposition 1.8.8 and Section 1.10]{Albeverio et al} and
\cite[Theorem 1.8.1-$(1)$]{Albeverio et al} we obtain the results
desired.
\end{proof}

\begin{remark}\label{Obs}
\begin{enumerate}
\item [(i)] If $f:\mathbb{Q}_{p}^{n}\rightarrow \mathbb{C}$ is a radial
function, then, as a consequence of the previous lemma and
\cite[Chapter VI-Section 1]{V-V-Z} we have that $f$ is a continuous
function on $\mathbb{Q}_{p}^{n}$.
\item [(ii)] If $f:\mathbb{Q}_{p}^{n}\rightarrow \mathbb{C}$ is a radial
function, then, it is clear that $|f|$ is a locally constant
function on $\mathbb{Q}_{p}^{n}$.
\end{enumerate}
\end{remark}

\begin{definition}\label{Hypothesis_A}[{\bf Hypothesis A}]
\label{Hypothesis A} Let $\boldsymbol{\psi}_{1},
\boldsymbol{\psi}_{2}:\mathbb{Q}_{p}^{n}\rightarrow \mathbb{C}$ be
functions. We say that $\boldsymbol{\psi}_{1}$ and
$\boldsymbol{\psi}_{2}$ satisfies the \textit{Hypothesis A} if the
following properties are met:
\begin{enumerate}
\item[(i)] $\boldsymbol{\psi}_{1}$ and
$\boldsymbol{\psi}_{2}$ are negative definite and radial
(consequently, continuous) functions with
$\boldsymbol{\psi}_{1}(||\xi||_{p})\neq 0$ for all $\xi \in
\mathbb{Q}_{p}^{n}$.
\item[(ii)] There is a ball $B_{r}^{n}$,
$r:=r(\boldsymbol{\psi}_{1},\boldsymbol{\psi}_{2})\in \mathbb{Z}$,
such that
\begin{equation*}
|\boldsymbol{\psi}_{1}(||\xi||_{p})|\geq
|\boldsymbol{\psi}_{2}(||\xi||_{p})| \text{ if and only if } \xi \in
B_{r}^{n}.
\end{equation*}
\end{enumerate}
\end{definition}

The condition $(ii)$ in the Definition \ref{Hypothesis_A} is
motivated by \cite[Chapter 1-Section I-$3$]{V-V-Z}. Denote by
$\mathbb{N}:=\left\{1,2,\ldots \right\}$ the set of natural numbers
and let $\mathbb{R}_{+}:=\left\{x\in \mathbb{R}:x\geq 0\right\}$.
Throughout this paper we will assume that $\boldsymbol{\psi}_{1}$
and $\boldsymbol{\psi}_{2}$ are functions satisfying the Hypothesis
A.

\begin{definition}\label{Def_Bessel_potential}
If $f\in \mathcal{D}'(\mathbb{Q}_{p}^{n})$, $\alpha \in \mathbb{C}$,
we define the non-archimedean generalized Bessel potential of order
$\alpha$ of $f$ by
\begin{equation*}
\mathcal{F}\left(\mathcal{A}^{\alpha}f\right)=\left[\max\{|\boldsymbol{\psi}_{1}(||\xi||_{p})|,|\boldsymbol{\psi}_{2}(||\xi||_{p})|\}\right]^{-\alpha}\widehat{f}.
\end{equation*}
\end{definition}

Let $\alpha\in\mathbb{C}$ fixed. For $\varphi \in
\mathcal{D}(\mathbb{Q}_{p}^{n})$ we define
\begin{eqnarray*}
(\mathcal{A}^{\alpha}\varphi)(x)&=&\mathcal{F}^{-1}_{\xi \rightarrow
x}\left(\left[\max\{|\boldsymbol{\psi}_{1}(||\xi||_{p})|,|\boldsymbol{\psi}_{2}(||\xi||_{p})|\}\right]^{-\alpha}\widehat{\varphi}(\xi)\right) \\
&=&\int_{\mathbf{\mathbb{Q}}_{p}^{n}}\chi_{p}(-x\cdot\xi)\left[\max\{|\boldsymbol{\psi}_{1}(||\xi||_{p})|,|\boldsymbol{\psi}_{2}(||\xi||_{p})|\}\right]^{-\alpha}\widehat{\varphi}(\xi)
d^{n}\xi, \text{ } x\in \mathbb{Q}_{p}^{n}.
\end{eqnarray*}

\begin{lemma} \label{operator_def} The application
\[
\begin{array}
[c]{cccc} \mathcal{A}^{\alpha}: & \mathcal{D}(\mathbb{Q}_{p}^{n})  &
\rightarrow & \mathcal{D}(\mathbb{Q}_{p}^{n}) \\
&  &  & \\
& \alpha & \longrightarrow & \mathcal{A}^{\alpha}\varphi
\end{array}
\]
corresponds to a well-defined $p$-adic pseudo-differential operator
where your symbol
$\left[\max\{|\boldsymbol{\psi}_{1}(||\xi||_{p})|,|\boldsymbol{\psi}_{2}(||\xi||_{p})|\}\right]^{-\alpha}\in$
{\LARGE$\varepsilon$}$(\mathbb{Q}_{p}^{n})$.
\end{lemma}

\begin{proof}
Note that
\begin{equation*}
\left[\max\{|\boldsymbol{\psi}_{1}(||\xi||_{p})|,|\boldsymbol{\psi}_{2}(||\xi||_{p})|\}\right]^{-\alpha}=|\boldsymbol{\psi}_{1}(||\xi||_{p})|^{-\alpha}1_{B_{r}^{n}}+|\boldsymbol{\psi}_{2}(||\xi||_{p})|^{-\alpha}1_{\mathbb{Q}_{p}^{n}\backslash
B_{r}^{n}}.
\end{equation*}
Given the fact that $1_{B_{r}^{n}}$ and
$1_{\mathbb{Q}_{p}^{n}\backslash B_{r}^{n}}$ assume the values $0$
and $1$ only, and in addition, as the set
{\LARGE$\varepsilon$}$(\mathbb{Q}_{p}^{n})$ is linear over the field
$\mathbb{C}$, see e.g. \cite[Chapter VI-Section 1]{V-V-Z}, then by
Lemma \ref{radial-locally constant}, Remark \ref{Obs}-$(ii)$ and
\cite[Example 5-p. 80]{V-V-Z} we have that the function
$\left[\max\{|\boldsymbol{\psi}_{1}(||\xi||_{p})|,|\boldsymbol{\psi}_{2}(||\xi||_{p})|\}\right]^{-\alpha}$
is locally constant on $\mathbb{Q}_{p}^{n}$.\\
Therefore, if $\varphi \in \mathcal{D}(\mathbb{Q}_{p}^{n})$ then by
\cite[Theorem 4.8.2]{V-V-Z} we have that
$\mathcal{A}^{\alpha}\varphi\in \mathcal{D}(\mathbb{Q}_{p}^{n})$.
\end{proof}
The function $\mathcal{A}^{\alpha}\varphi$, $\alpha \in \mathbb{C}$,
is called the generalized Bessel potential of order $\alpha$ of
$\varphi$.

\section{\label{Section_convolution_kernel} Convolution kernels attached to generalized
Bessel potentials and its applications}

In this section we study certain properties and applications
corresponding to the convolution kernels attached to generalized
Bessel potentials. From now on, $\alpha$ is a real number such that
$\alpha\in \mathbb{R}_{+}\backslash \left\{0 \right\}$.\\
We begin with the following definition.

\begin{definition}\label{convolution_kernel}
We define the convolution kernel $K_{\alpha}$ of the generalized
Bessel potential $\mathcal{A}^{\alpha}$ by
\begin{equation}\label{Definition_K_a}
K_{\alpha}(x):=\int_{\mathbf{\mathbb{Q}}_{p}^{n}}\chi_{p}(-x\cdot\xi)\left[\max\{|\boldsymbol{\psi}_{1}(||\xi||_{p})|,|\boldsymbol{\psi}_{2}(||\xi||_{p})|\}\right]^{-\alpha}d^{n}\xi,
\text{  } x \in \mathbb{Q}_{p}^{n}.
\end{equation}
\end{definition}

\begin{remark}\label{Rem_convolution_kernel}
Note that the symbol
$\left[\max\{|\boldsymbol{\psi}_{1}(||\xi||_{p})|,|\boldsymbol{\psi}_{2}(||\xi||_{p})|\}\right]^{-\alpha}$
of the pseudo-differential operator $\mathcal{A}^{\alpha}$ defines a
regular distribution on $\mathbb{Q}_{p}^{n}$. In this case, by
\cite[Proposition 4.9.1]{Albeverio et al} we have that
$K_{\alpha}\in \mathcal{D}'(\mathbb{Q}_{p}^{n})$.
\end{remark}

\begin{theorem}\label{Theorem_Convolution_Kernel}
The convolution kernel $K_{\alpha}$ satisfies the following
conditions:
\begin{enumerate}
\item[(i)] $K_{\alpha}\ast \varphi=\mathcal{A}^{\alpha} \varphi$, for all
$\varphi\in \mathcal{D}(\mathbb{Q}_{p}^{n})$.
\item[(ii)] For $x\in \mathbb{Q}_{p}^{n}\backslash \left\{0\right\}$, we have
that
\begin{equation*}
\begin{split}
K_{\alpha}(x)&=||x||_{p}^{-n}\left\{(1-p^{-n})\sum_{j=0}^{\infty}\left(\Bigg.\left[\max\{|\boldsymbol{\psi}_{1}(||x||_{p}^{-1}p^{-j})|,|\boldsymbol{\psi}_{2}(||
x||_{p}^{-1}p^{-j})|\}\right]^{-\alpha}\right.\right.\\
&\quad
\Bigg.-\left[\max\{|\boldsymbol{\psi}_{1}(||x||_{p}^{-1}p)|,|\boldsymbol{\psi}_{2}(||x||_{p}^{-1}p)|\}\right]^{-\alpha}\Bigg)
p^{-nj} \Bigg\}.
\end{split}
\end{equation*}
\end{enumerate}
\end{theorem}

\begin{proof}
\begin{enumerate}
\item[(i)] By using Fubini's Theorem, for $\varphi\in \mathcal{D}(\mathbb{Q}_{p}^{n})$ and $x\in
\mathbb{Q}_{p}^{n}$ we have that
\begin{eqnarray*}
(\mathcal{A}^{\alpha}\varphi)(x)&=&\int_{\mathbf{\mathbb{Q}}_{p}^{n}}\int_{\mathbf{\mathbb{Q}}_{p}^{n}}\chi_{p}(y-x\cdot\xi)\left[\max\{|\boldsymbol{\psi}_{1}(||\xi||_{p})|,|\boldsymbol{\psi}_{2}(||\xi||_{p})|\}\right]^{-\alpha}\varphi(y)d^{n}yd^{n}\xi\\
&=&\int_{\mathbf{\mathbb{Q}}_{p}^{n}}\int_{\mathbf{\mathbb{Q}}_{p}^{n}}\overline{\chi_{p}(x-y\cdot\xi)}\left[\max\{|\boldsymbol{\psi}_{1}(||\xi||_{p})|,|\boldsymbol{\psi}_{2}(||\xi||_{p})|\}\right]^{-\alpha}\varphi(y)d^{n}yd^{n}\xi\\
&=&\int_{\mathbf{\mathbb{Q}}_{p}^{n}}\int_{\mathbf{\mathbb{Q}}_{p}^{n}}\overline{\chi_{p}(z\cdot\xi)}\left[\max\{|\boldsymbol{\psi}_{1}(||\xi||_{p})|,|\boldsymbol{\psi}_{2}(||\xi||_{p})|\}\right]^{-\alpha}\varphi(x-z)d^{n}zd^{n}\xi\\
&=&\int_{\mathbf{\mathbb{Q}}_{p}^{n}}\int_{\mathbf{\mathbb{Q}}_{p}^{n}}\chi_{p}(-z\cdot\xi)\left[\max\{|\boldsymbol{\psi}_{1}(||\xi||_{p})|,|\boldsymbol{\psi}_{2}(||\xi||_{p})|\}\right]^{-\alpha}d^{n}\xi\varphi(x-z)d^{n}z\\
&=&\int_{\mathbf{\mathbb{Q}}_{p}^{n}}K_{\alpha}(z)\varphi(x-z)d^{n}z\\
&=&(K_{\alpha}\ast\varphi)(x).
\end{eqnarray*}
\item[(ii)] Let $x=p^{\gamma}x_{0}\neq 0$ with $\gamma \in \mathbb{Z}$ and
$||x_{0}||_{p}=1$. Using the changes of variables $w=p^{\gamma}\xi$
and $z=p^{j}w$, respectively, we have that
\begin{eqnarray*}
K_{\alpha}(x)&=&\int_{\mathbf{\mathbb{Q}}_{p}^{n}}\chi_{p}(-p^{\gamma}\xi\cdot
x_{0})\left[\max\{|\boldsymbol{\psi}_{1}(||\xi||_{p})|,|\boldsymbol{\psi}_{2}(||\xi||_{p})|\}\right]^{-\alpha}d^{n}\xi\\
&=&||x||_{p}^{-n}\hspace{-0.4cm}{\sum_{-\infty<j<\infty}}\hspace{-0.3cm}{\left[\max\{|\boldsymbol{\psi}_{1}(p^{\gamma+j})|,|\boldsymbol{\psi}_{2}(p^{\gamma+j})|\}\right]^{-\alpha}}\hspace{-0.4cm}{\int\limits_{||w||_{p}=p^{j}}\hspace{-0.3cm}{\chi_{p}(-w\cdot
x_{0})}d^{n}w}\\
&=&||x||_{p}^{-n}\hspace{-0.4cm}{\sum_{-\infty<j<\infty}}\hspace{-0.3cm}{\left[\max\{|\boldsymbol{\psi}_{1}(p^{\gamma+j})|,|\boldsymbol{\psi}_{2}(p^{\gamma+j})|\}\right]^{-\alpha}}\hspace{-0.4cm}{\int\limits_{||p^{j}w||_{p}=1}\hspace{-0.3cm}{\chi_{p}(-w\cdot
x_{0})}d^{n}w}\\
&=&||x||_{p}^{-n}\hspace{-0.4cm}{\sum_{-\infty<j<\infty}}\hspace{-0.3cm}{\left[\max\{|\boldsymbol{\psi}_{1}(p^{\gamma+j})|,|\boldsymbol{\psi}_{2}(p^{\gamma+j})|\}\right]^{-\alpha}}\hspace{-0.4cm}{\int\limits_{||z||_{p}=1}\hspace{-0.3cm}{\chi_{p}(-p^{-j}x_{0}\cdot
z)}d^{n}z}
\end{eqnarray*}
By using the formula
\begin{equation} \label{formula1}
\int_{||z||_{p}=1}\chi_{p}\left(-p^{-j}x_{0}\cdot z\right)
d^{n}z=\left\{
\begin{array}{lll}
1-p^{-n}, & \text{if} & \text{ }j\leq 0 , \\
-p^{-n}, & \text{if} & \text{ }j=1, \\
0, & \text{if} & \text{ } j\geq 2,
\end{array}
\right.
\end{equation}
we have that
\begin{equation}\label{K_alpha=}
\begin{split}
K_{\alpha}(x)&=||x||_{p}^{-n}\left\{(1-p^{-n})\sum_{j=0}^{\infty}p^{-nj}\left[\max\{|\boldsymbol{\psi}_{1}(||x||_{p}^{-1}p^{-j})|,|\boldsymbol{\psi}_{2}(||
x||_{p}^{-1}p^{-j})|\}\right]^{-\alpha}\right.\\
&\quad
\Bigg.-\left[\max\{|\boldsymbol{\psi}_{1}(||x||_{p}^{-1}p)|,|\boldsymbol{\psi}_{2}(||x||_{p}^{-1}p)|\}\right]^{-\alpha}
\Bigg\}.
\end{split}
\end{equation}
Since $(1-p^{-n})\sum_{j=0}^{\infty}p^{-nj}=1$, then by
(\ref{K_alpha=}) the desired is obtained.
\end{enumerate}
\end{proof}

As a direct consequence of Theorem
\ref{Theorem_Convolution_Kernel}-$(ii)$ and the fact that
$K_{\alpha}(0)>0$, we have the following Corollary.

\begin{corollary} \label{Corollary_1}
If the function
$\left[\max\{|\boldsymbol{\psi}_{1}(||x||_{p})|,|\boldsymbol{\psi}_{2}(||
x||_{p}^{-1})|\}\right]^{-\alpha}$ is increasing with respect to
$||\cdot||_{p}$, then $K_{\alpha}(x)\geq 0$ for all $x\in
\mathbb{Q}_{p}^{n}$.
\end{corollary}

\begin{example}
Let $\boldsymbol{\psi}_{1}$ and $\boldsymbol{\psi}_{2}$ be functions
given by $\boldsymbol{\psi}_{1}(\xi)=1$ and
$\boldsymbol{\psi}_{2}(\xi)=||\xi||_{p}^{\beta}$, fixed
$\beta>\frac{n}{\alpha}$, for all $\xi\in \mathbb{Q}_{p}^{n}$. Then,
by Example \ref{example neg def}, we have that the functions
$\boldsymbol{\psi}_{1}$ and $\boldsymbol{\psi}_{2}$ are negative
definite radial functions on
$\mathbb{Q}_{p}^{n}$.\\
On the other hand, it is easy to check that
$\boldsymbol{\psi}_{1}(||\xi||_{p})\geq
\boldsymbol{\psi}_{2}(||\xi||_{p})$ if and
only if $\xi \in \mathbb{Z}_{p}^{n}$.\\
By defining for $\alpha>0$
\begin{equation*}
K_{\alpha}(x):=\int_{\mathbf{\mathbb{Q}}_{p}^{n}}\chi_{p}(-x\cdot\xi)\left[\max\{1,||\xi||_{p}^{\beta}\}\right]^{-\alpha}d^{n}\xi,
\text{  } x \in \mathbb{Q}_{p}^{n},
\end{equation*}
by a direct calculation one verifies that $K_{\alpha}(x)\geq 0$, for
all $x\in \mathbb{Q}_{p}^{n}$.
\end{example}

Next we will obtain some relevant applications corresponding to the
convolution kernels $K_{\alpha}$, $\alpha>0$.

\begin{definition}
A function $f:\mathbb{Q}_{p}^{n}\rightarrow\mathbb{C}$ is called
positive definite, if
\begin{equation*}
\sum\nolimits_{i,j=1}^{m}f(x_{i}-x_{j})\lambda _{i}\overline{\lambda
_{j}}\geq 0
\end{equation*}
for all $m\in \mathbb{N}$ , $x_{1},\ldots ,x_{m}$ $\in $
$\mathbb{Q}_{p}^{n}$ and $\lambda _{1},\ldots ,\lambda _{m}\in
\mathbb{C}$.
\end{definition}

\begin{lemma}\label{positive_definite} The function
$\left[\max\{|\boldsymbol{\psi}_{1}(||\xi||_{p})|,|\boldsymbol{\psi}_{2}(||\xi||_{p})|\}\right]^{-\alpha}:\mathbb{Q}_{p}^{n}\rightarrow\mathbb{R}_{+}\backslash
\left\{0 \right\}$ is positive definite.
\end{lemma}

\begin{proof}
We first note that
\begin{eqnarray*}
\left[\max\{|\boldsymbol{\psi}_{1}(||x-y||_{p})|,|\boldsymbol{\psi}_{2}(||x-y||_{p})|\}\right]^{-\alpha}\hspace{-0.4cm}&=&\hspace{-0.3cm}\left[\max\{|\boldsymbol{\psi}_{1}(||y-x||_{p})|,|\boldsymbol{\psi}_{2}(||y-x||_{p})|\}\right]^{-\alpha}\\
&=&\hspace{-0.3cm}\left[\max\{|\boldsymbol{\psi}_{2}(||y-x||_{p})|,|\boldsymbol{\psi}_{1}(||y-x||_{p})|\}\right]^{-\alpha}
\end{eqnarray*}
for all $x, y \in \mathbb{Q}_{p}^{n}$.\\
Let $c_{1}, c_{2},\ldots, c_{m}$ be integer numbers and $\xi_{1},
\xi_{2},\ldots,\xi_{m}\in \mathbb{Q}_{p}^{n}$, $m\in \mathbb{N}$.
Then, by a direct calculation one verifies that
\begin{equation*}
\sum\nolimits_{i,j=1}^{m}\left[\max\{|\boldsymbol{\psi}_{1}(||\xi_{i}-\xi_{j}||_{p})|,|\boldsymbol{\psi}_{2}(||\xi_{i}-\xi_{j}||_{p})|\}\right]^{-\alpha}c_{i}c_{j}
\end{equation*}
is exactly
$\sum\nolimits_{i,j=1}^{m}\left[\max\{|\boldsymbol{\psi}_{1}(||\xi_{i}-\xi_{j}||_{p})|,|\boldsymbol{\psi}_{2}(||\xi_{i}-\xi_{j}||_{p})|\}\right]^{-\alpha}c_{i}^{2}\geq
0$.\\
On the other hand, it is clear that
\begin{equation*}
\left[\max\{|\boldsymbol{\psi}_{1}(||-x||_{p})|,|\boldsymbol{\psi}_{2}(||-x||_{p})|\}\right]^{-\alpha}=\left[\max\{|\boldsymbol{\psi}_{1}(||x||_{p})|,|\boldsymbol{\psi}_{2}(||x||_{p})|\}\right]^{-\alpha},
\end{equation*}
for all $x\in \mathbb{Q}_{p}^{n}$.\\
Therefore, the desired result follow from \cite[Exercise 3.7-p.
13]{Berg-Gunnar}.
\end{proof}

\begin{lemma}\label{lemma_measures} $K_{\alpha}$ is a positive bounded measure on
$\mathbb{Q}_{p}^{n}$.
\end{lemma}

\begin{proof}
By \cite[Proposition 4.9.1]{Albeverio et al}, Remark
\ref{Rem_convolution_kernel} and (\ref{Definition_K_a}) we have that
\begin{equation*}
\mathcal{F}_{x\rightarrow
\xi}\left(K_{\alpha}\right)=\left[\max\{|\boldsymbol{\psi}_{1}(||x||_{p})|,|\boldsymbol{\psi}_{2}(||x||_{p})|\}\right]^{-\alpha}\in
\mathcal{D}'(\mathbb{Q}_{p}^{n}).
\end{equation*}
The result follow from Lemma \ref{positive_definite}, \cite[Section
4.9]{Albeverio et al} and \cite[Theorem 3.12]{Berg-Gunnar}.
\end{proof}

\begin{theorem}\label{convolution_semigroup}
Suppose that $|\boldsymbol{\psi}_{1}(0)|\geq 1$. Then, the family
$\left(K_{\alpha}\right)_{\alpha>0}$ determine a convolution
semigroup on $\mathbb{Q}_{p}^{n}$, i.e.
$\left(K_{\alpha}\right)_{\alpha>0}$ satisfies the following
properties:
\begin{enumerate}
\item[(i)] For all $\alpha>0$, $K_{\alpha}$ is a positive bounded measure
on $\mathbb{Q}_{p}^{n}$.
\item[(ii)] For all $\alpha>0$, $K_{\alpha}(\mathbb{Q}_{p}^{n})\leq
1$.
\item[(iii)] For all $\alpha_{1}, \alpha_{2}>0$ we have that $K_{\alpha_{1}}\ast
K_{\alpha_{2}}=K_{\alpha_{1}+\alpha_{2}}$.
\item[(iv)] $\lim_{\alpha\rightarrow 0}K_{\alpha}=\delta$, where $\delta$ is the Dirac delta function.
\end{enumerate}
\end{theorem}

\begin{proof}
\item[(i)] The result follows from Lemma \ref{lemma_measures}.
\item[(ii)] By Lemma \ref{lemma_measures} and \cite[$(7)$-p.
14]{Berg-Gunnar} we have that
$(K_{\alpha})(\mathbb{Q}_{p}^{n})=\frac{1}{|\boldsymbol{\psi}_{1}(0)|^{\alpha}}$.
Therefore, if $|\boldsymbol{\psi}_{1}(0)|\geq 1$, then
$(K_{\alpha})(\mathbb{Q}_{p}^{n})\leq 1$.
\item[(iii)] Let $\alpha_{1}$ and $\alpha_{2}$ be real numbers such that $\alpha_{1},
\alpha_{2}>0$. Then, by \cite[Proposition 4.9.1]{Albeverio et al},
Remark \ref{Rem_convolution_kernel}, Lemma \ref{operator_def} and
(\ref{Definition_K_a}) we have for $i=1,2$,
$\mathcal{F}_{x\rightarrow
\xi}\left(K_{\alpha_{i}}\right)=\left[\max\{|\boldsymbol{\psi}_{1}(||x||_{p})|,|\boldsymbol{\psi}_{2}(||x||_{p})|\}\right]^{-\alpha_{i}}\in$
{\LARGE$\varepsilon$}$(\mathbb{Q}_{p}^{n})\bigcap\mathcal{D}'(\mathbb{Q}_{p}^{n})$.
Moreover, by \cite[Theorem 4.9.3]{Albeverio et al} and \cite[Lemma
4.7.2]{Albeverio et al}, respectively, there exist integers
$N_{i}:=N_{i}(\alpha_{i})$, such that $supp(K_{\alpha_{i}})\subset
B_{N_{i}}^{n}$, $i=1,2$, and the convolution $K{\alpha_{1}}\ast
K{\alpha_{2}}$ exists.\\
On the other hand, by the $n$-dimensional version of
\cite[Theorem-p. 115]{V-V-Z} and \cite[Examples 1 and 2- p. 113 y
114, respectively]{V-V-Z}, we have that the product
$\left[\max\{|\boldsymbol{\psi}_{1}|,|\boldsymbol{\psi}_{2}|\}\right]^{-\alpha_{1}}\left[\max\{|\boldsymbol{\psi}_{1}|,|\boldsymbol{\psi}_{2}|\}\right]^{-\alpha_{2}}$
exists and
\begin{equation*}
\mathcal{F}^{-1}_{\xi\rightarrow
x}\left(\left[\max\{|\boldsymbol{\psi}_{1}|,|\boldsymbol{\psi}_{2}|\}\right]^{-\alpha_{1}}\left[\max\{|\boldsymbol{\psi}_{1}|,|\boldsymbol{\psi}_{2}|\}\right]^{-\alpha_{2}}
\right)=K_{\alpha_{1}}\ast K_{\alpha_{2}}.
\end{equation*}
Therefore, for $\varphi \in \mathcal{D}(\mathbb{Q}_{p}^{n})$ we have
that
\begin{eqnarray*}
\left\langle K_{\alpha_{1}+\alpha_{2}},\varphi
\right\rangle&=&\left\langle \mathcal{F}^{-1}_{\xi\rightarrow
x}\left(\left[\max\{|\boldsymbol{\psi}_{1}|,|\boldsymbol{\psi}_{2}|\}\right]^{-(\alpha_{1}+\alpha_{2})}\right),\varphi
\right\rangle\\
&=&\left\langle\mathcal{F}^{-1}_{\xi\rightarrow
x}\left(\left[\max\{|\boldsymbol{\psi}_{1}|,|\boldsymbol{\psi}_{2}|\}\right]^{-\alpha_{1}}\left[\max\{|\boldsymbol{\psi}_{1}|,|\boldsymbol{\psi}_{2}|\}\right]^{-\alpha_{2}}
\right),\varphi\right\rangle\\
&=&\left\langle K_{\alpha_{1}}\ast
K_{\alpha_{2}},\varphi\right\rangle.
\end{eqnarray*}
\item[(iv)] The result follows from Remark
\ref{Rem_convolution_kernel} and \cite[Example 9-p. 44 and Chapter
1-Section VI]{V-V-Z}.
\end{proof}

As a direct consequence of above theorem, we have the following
corollary.

\begin{corollary} \label{Corollary_probability_measure} If $|\boldsymbol{\psi}_{1}(0)|=1$, then
$K_{\alpha}$, $\alpha>0$, is a probability measure on
$\mathbb{Q}_{p}^{n}$.
\end{corollary}

\section{\label{Green and Heat_Kernel} The Green function and the heat Kernel}

In this section we will study certain properties corresponding to
the Green function and the heat Kernel attached to operator
$\mathcal{A}^{\alpha}$.

\subsection{The Green function }
The Dirac $\delta$-function is defined by
\begin{equation*}
\left<\delta,\varphi \right>=\varphi(0), \text{  } \forall\varphi\in
\mathcal{D}(\mathbb{Q}_{p}^{n}).
\end{equation*}
It is clear that $\delta \in \mathcal{D}'(\mathbb{Q}_{p}^{n})$ and
$\delta(x)=0$ for all $x\neq 0$, i.e. $supp(\delta)=\left\{0
\right\}$, see \cite{Albeverio et al}, \cite{V-V-Z} for details.
\begin{definition}
Let $G:=G(m,\alpha)\in \mathcal{D}'(\mathbb{Q}_{p}^{n})$ be a
distribution satisfying the equation
\begin{equation} \label{Def_Green}
(m^{2}+\mathcal{A}^{\alpha})G=\delta, \ \ \text{ } m\in
\mathbb{R}_{+}\backslash \left\{0\right\}.
\end{equation}
Then $G$ is called a Green function of the operator
$\mathcal{A}^{\alpha}$.
\end{definition}

\begin{lemma}\label{Green_function}
The distribution
\begin{equation*}
G(x):=\mathcal{F}^{-1}_{\xi \rightarrow
x}\left(\frac{1}{m^{2}+\left[\max\{|\boldsymbol{\psi}_{1}(||\xi||_{p})|,|\boldsymbol{\psi}_{2}(||\xi||_{p})|\}\right]^{-\alpha}
}\right), \text{  } x\in \mathbf{\mathbb{Q}}_{p}^{n},
\end{equation*}
is the Green function of the operator $\mathcal{A}^{\alpha}$.
\end{lemma}

\begin{proof}
Since $G\in L_{loc}^{1}(\mathbb{Q}_p^{n})$, then $G$ determine a
regular distribution by the formula
\begin{equation*}
\left<G,\varphi\right> =\int_{\mathbb{Q}_{p}^{n}}G(x)
\varphi(x)d^{n}x, \text{ } \varphi\in
\mathcal{D}(\mathbb{Q}_{p}^{n}).
\end{equation*}
Then, for any $\varphi \in \mathcal{D}(\mathbb{Q}_{p}^{n})$ we have
that
\begin{align*}
&\left<(m^{2}+\mathcal{A}^{\alpha})G,\varphi\right>\\
&=\hspace{-0.1cm}\left<G,(m^{2}+\mathcal{A}^{\alpha})\varphi\right>\\
&=\hspace{-0.1cm}\left<\mathcal{F}^{-1}_{\xi \rightarrow
x}\left(\frac{1}{m^{2}+\left[\max\{|\boldsymbol{\psi}_{1}(||\xi||_{p})|,|\boldsymbol{\psi}_{2}(||\xi||_{p})|\}\right]^{-\alpha}
}\right),(m^{2}+\mathcal{A}^{\alpha})\varphi\right>\\
&=\hspace{-0.1cm}\left<\hspace{-0.1cm}\frac{1}{m^{2}+\left[\max\{|\boldsymbol{\psi}_{1}(||\xi||_{p})|,|\boldsymbol{\psi}_{2}(||\xi||_{p})|\}\right]^{-\alpha}
},\hspace{-0.1cm}\left(\hspace{-0.1cm}\left[\max\{|\boldsymbol{\psi}_{1}(||\xi||_{p})|,\hspace{-0.1cm}|\boldsymbol{\psi}_{2}(||\xi||_{p})|\}\right]^{-\alpha}\hspace{-0.2cm}+m^{2}\right)\hspace{-0.1cm}\widehat{\varphi}\hspace{-0.1cm}\right>\\
&=<1,\widehat{\varphi}>\\
&=<\delta,\varphi>
\end{align*}
All the above shows that
\begin{equation*}
(m^{2}+\mathcal{A}^{\alpha})G=\delta, \ \ \text{ } m\in
\mathbb{R}_{+}\backslash \left\{0\right\}.
\end{equation*}
Therefore, $G$ is the Green function of the operator
$\mathcal{A}^{\alpha}$.
\end{proof}

\begin{theorem}\label{properties_Green}
The Green function $G$ satisfies the following properties:
\begin{enumerate}
\item[(i)] For all $x\in \mathbb{Q}_{p}^{n}\backslash \left\{0\right\}$ we have that
\begin{equation*}
\begin{split}
&G(x)=||x||_{p}^{-n}\left\{(1-p^{-n})\sum_{j=0}^{\infty}\frac{p^{-nj}}{m^{2}+\left[\max\{|\boldsymbol{\psi}_{1}(p^{\gamma-j})|,|\boldsymbol{\psi}_{2}(p^{\gamma-j})|\}\right]^{-\alpha}}\right.\\
&\quad
\Bigg.-\frac{1}{m^{2}+\left[\max\{|\boldsymbol{\psi}_{1}(p^{\gamma+1})|,|\boldsymbol{\psi}_{2}(p^{\gamma+1})|\}\right]^{-\alpha}}
\Bigg\}.
\end{split}
\end{equation*}
\item[(ii)] If $\left[\max\{|\boldsymbol{\psi}_{1}(||\xi||_{p})|,|\boldsymbol{\psi}_{2}(||\xi||_{p})|\}\right]\geq
1$ for all $\xi \in \mathbb{Q}_{p}^{n}$, then, there exist positive
real constants $K_{1}:=\frac{1}{m^{2}(m^{2}+1)}$ and
$K_{2}:=\frac{1}{m^{2}}$ such that
\begin{equation*}
-K_{1}||x||_{p}^{-n}\leq G(x)\leq K_{2}||x||_{p}^{-n}, \text{ for
all } x\in \mathbb{Q}_{p}^{n}\backslash \left\{0\right\},
\end{equation*}
\item[(iii)] $G(x)$ is a real-valued, radial and continuous
function.
\end{enumerate}
\end{theorem}

\begin{proof}
\begin{enumerate}
\item[(i)] Let $x=p^{\gamma}x_{0}\neq 0$ with $\gamma \in \mathbb{Z}$ and
$||x_{0}||_{p}=1$. Then,
\begin{equation*}
G(x)=\int_{\mathbf{\mathbb{Q}}_{p}^{n}}\chi_{p}(-p^{\gamma}\xi\cdot
x_{0})\left(\frac{1}{m^{2}+\left[\max\{|\boldsymbol{\psi}_{1}(||\xi||_{p})|,|\boldsymbol{\psi}_{2}(||\xi||_{p})|\}\right]^{-\alpha}}\right)d^{n}\xi.
\end{equation*}
We make a change of variables, namely $z=p^{\gamma}\xi$, we have
that $G(x)$ is exactly
\begin{align*}
&||x||_{p}^{-n}\hspace{-0.4cm}\sum_{-\infty<j<\infty}\hspace{-0.2cm}\left(\frac{1}{m^{2}+\left[\max\{|\boldsymbol{\psi}_{1}(p^{\gamma}||z||_{p})|,|\boldsymbol{\psi}_{2}(p^{\gamma}||z||_{p})|\}\right]^{-\alpha}}\right)\hspace{-0.2cm}\int\limits_{||z||_{p}=p^{j}}\hspace{-0.3cm}\chi_{p}(-z\cdot
x_{0})d^{n}z\\
&=||x||_{p}^{-n}\hspace{-0.4cm}\sum_{-\infty<j<\infty}\hspace{-0.2cm}\left(\frac{1}{m^{2}+\left[\max\{|\boldsymbol{\psi}_{1}(p^{\gamma+j})|,|\boldsymbol{\psi}_{2}(p^{\gamma+j})|\}\right]^{-\alpha}}\right)\hspace{-0.2cm}\int\limits_{||p^{j}z||_{p}=1}\hspace{-0.3cm}\chi_{p}(-z\cdot
x_{0})d^{n}z\\
&=||x||_{p}^{-n}\hspace{-0.4cm}\sum_{-\infty<j<\infty}\hspace{-0.2cm}\left(\frac{p^{nj}}{m^{2}+\left[\max\{|\boldsymbol{\psi}_{1}(p^{\gamma+j})|,|\boldsymbol{\psi}_{2}(p^{\gamma+j})|\}\right]^{-\alpha}}\right)\hspace{-0.2cm}\int\limits_{||z||_{p}=1}\hspace{-0.3cm}\chi_{p}(-p^{-j}x_{0}\cdot
z)d^{n}z
\end{align*}
By using the formula (\ref{formula1}), we obtain the desired
equality.
\item[(ii)] By $(i)$ and taking
$K_{2}:=\frac{1}{m^{2}}$ we have that
\begin{eqnarray*}
G(x)&\leq& ||x||_{p}^{-n}\left\{(1-p^{-n})\sum_{j=0}^{\infty}\left(\frac{p^{-nj}}{m^{2}+\left[\max\{|\boldsymbol{\psi}_{1}(p^{\gamma-j})|,|\boldsymbol{\psi}_{2}(p^{\gamma-j})|\}\right]^{-\alpha}}\right)\right\}\\
&\leq&||x||_{p}^{-n}\left\{\frac{1}{m^{2}}\sum_{j=0}^{\infty}\left(p^{-nj}-p^{-n(j+1)}\right)\right\}\\
&=& K_{2}||x||_{p}^{-n}.
\end{eqnarray*}
On the other hand, by (i) and taking
$K_{1}:=\frac{1}{m^{2}(m^{2}+1)}$ we have that
\begin{eqnarray*}
G(x)&\geq & ||x||_{p}^{-n}\left\{(1-p^{-n})\sum_{j=0}^{\infty}\left(\frac{p^{-nj}}{m^{2}+1}\right)-\frac{1}{m^{2}+\left[\max\{|\boldsymbol{\psi}_{1}(p^{\gamma+1})|,|\boldsymbol{\psi}_{2}(p^{\gamma+1})|\}\right]^{-\alpha}}\right\}\\\
&\geq&||x||_{p}^{-n}\left\{\frac{1}{m^{2}+1}\sum_{j=0}^{\infty}\left(p^{-nj}-p^{-n(j+1)}\right)-\frac{1}{m^{2}}\right\}\\
&=&-K_{1}||x||_{p}^{-n}.
\end{eqnarray*}
\item[(iii)] Is a direct consequence of $(i)$, Lemma
\ref{Green_function} and Remark \ref{Obs}-$(i)$.
\end{enumerate}
\end{proof}

\subsection{The heat kernel }
We assume throughout this section that $|\boldsymbol{\psi}_{1}|$ and
$|\boldsymbol{\psi}_{2}|$ are increasing functions with respect to
$||\cdot||_{p}$.

We define the heat Kernel (also called fundamental solution)
attached to operator $\mathcal{A}^{\alpha}$ as
\begin{equation}
Z_{t}(x)=Z(x,t):=\int_{\mathbb{Q}_{p}^{n}}\chi_{p}\left(-x\cdot
\xi\right)e^{-t\left[\max\{|\boldsymbol{\psi}_{1}(||x||_{p})|,|\boldsymbol{\psi}_{2}(||x||_{p})|\}\right]^{-\alpha}}
d^{n}\xi, \label{Def_Z}
\end{equation}
for $x\in\mathbb{Q}_{p}^{n}$ and $t\geq 0$.

\begin{remark}\label{Obs_Z}
 If $t=0$, then by
\cite[Example 4.9.1]{Albeverio et al} we have that
\begin{equation*}
Z_{0}(x)=\int_{\mathbb{Q}_{p}^{n}}\chi_{p}\left(-x\cdot
\xi\right)d^{n}\xi=\delta \in
\mathcal{D}^{\prime}(\mathbb{Q}_{p}^{n}).
\end{equation*}
For $t>0$ and by a direct calculation one verifies that
$e^{-t\left[\max\{|\boldsymbol{\psi}_{1}|,|\boldsymbol{\psi}_{2}|\}\right]^{-\alpha}}\notin
L^{1}(\mathbb{Q}_{p}^{n})$. Now, by using the fact that
$e^{-t\left[\max\{|\boldsymbol{\psi}_{1}|,|\boldsymbol{\psi}_{2}|\}\right]^{-\alpha}}\in
L^{1}_{loc}(\mathbb{Q}_{p}^{n})$, then
$e^{-t\left[\max\{|\boldsymbol{\psi}_{1}|,|\boldsymbol{\psi}_{2}|\}\right]^{-\alpha}}$
defines a regular distribution. Therefore, by \cite[Proposition
4.9.1]{Albeverio et al} we have that $Z_{t}(x)\in
\mathcal{D}^{\prime}(\mathbb{Q}_{p}^{n})$.
\end{remark}

\begin{theorem}\label{theorem_Z}
$Z_{t}(x)\leq 0$ for all $x\in \mathbb{Q}_{p}^{n}\backslash
\left\{0\right\}$ and $t>0$.
\end{theorem}

\begin{proof}
Let $x=p^{\gamma}x_{0}\neq 0$ such that $\gamma\in \mathbb{Z}$ and
$||x_{0}||_{p}=1$. Then, by (\ref{Def_Z}) and using the changes of
variables $w=p^{\gamma}\xi$ and $z=p^{j}w$, respectively, we have
that
\begin{eqnarray*}
Z(x,t)&=&||x||_{p}^{-n}\int_{\mathbf{\mathbb{Q}}_{p}^{n}}\chi_{p}(-w\cdot x_{0})e^{-t\left[\max\{|\boldsymbol{\psi}_{1}(||w||_{p}p^{\gamma})|,|\boldsymbol{\psi}_{2}(||w||_{p}p^{\gamma})|\}\right]^{-\alpha}}d^{n}w\\
&=&||x||_{p}^{-n}\sum_{-\infty<j<\infty}e^{-t\left[\max\{|\boldsymbol{\psi}_{1}(p^{\gamma+j})|,|\boldsymbol{\psi}_{2}(p^{\gamma+j})|\}\right]^{-\alpha}}\hspace{-0.5cm}\int\limits_{||p^{j}w||_{p}=1}\hspace{-0.2cm}\chi_{p}(-w\cdot x_{0})d^{n}w\\
&=&||x||_{p}^{-n}\hspace{-0.2cm}\sum_{-\infty<j<\infty}e^{-t\left[\max\{|\boldsymbol{\psi}_{1}(p^{\gamma+j})|,|\boldsymbol{\psi}_{2}(p^{\gamma+j})|\}\right]^{-\alpha}}p^{nj}\hspace{-0.5cm}\int\limits_{||z||_{p}=1}\hspace{-0.2cm}\chi_{p}(-p^{-j}x_{0}\cdot
z)d^{n}z.
\end{eqnarray*}
By using the formula (\ref{formula1}) one verifies that
\begin{equation} \label{expansion_Z}
\begin{split}
Z(x,t)&=||x||_{p}^{-n}\left\{(1-p^{-n})\sum_{j=0}^{\infty}\Bigg.e^{-t\left[\max\{|\boldsymbol{\psi}_{1}(||x||_{p}^{-1}p^{-j})|,|\boldsymbol{\psi}_{2}(||x||_{p}^{-1}p^{-j})|\}\right]^{-\alpha}}p^{-nj}\right.\\
&\quad
\Bigg.-e^{-t\left[\max\{|\boldsymbol{\psi}_{1}(||x||_{p}^{-1}p)|,|\boldsymbol{\psi}_{2}(||x||_{p}^{-1}p)|\}\right]^{-\alpha}}\Bigg\}.
\end{split}
\end{equation}
Now, we need consider three cases for $||x||_{p}$.\\
\case $||x||_{p}>p^{-r}$. In this case, we have that
$||x||_{p}^{-1}<p^{r}$. Therefore, $||x||_{p}^{-1}p^{-j}\leq p^{r}$
for all $j\geq -1$. Then, by (\ref{expansion_Z}) and Definition
\ref{Hypothesis_A} we have that
\begin{eqnarray*}
Z(x,t)&=&||x||_{p}^{-n}\left\{(1-p^{-n})\sum_{j=0}^{\infty}e^{-t|\boldsymbol{\psi}_{1}(||x||_{p}^{-1}p^{-j})|^{-\alpha}}p^{-nj}-e^{-t|\boldsymbol{\psi}_{1}(||x||_{p}^{-1}p)|^{-\alpha}}
\right\}\\
&\leq&
||x||_{p}^{-n}\left\{e^{-t|\boldsymbol{\psi}_{1}(||x||_{p}^{-1})|^{-\alpha}}-e^{-t|\boldsymbol{\psi}_{1}(||x||_{p}^{-1}p)|^{-\alpha}}\right\}\\
&\leq& 0.
\end{eqnarray*}
\case $||x||_{p}=p^{-r}$. In this case, we have that
$||x||_{p}^{-1}=p^{r}$. Therefore, $||x||_{p}^{-1}p^{-j}\leq p^{r}$
for all $j\geq 0$ and moreover, $||x||_{p}^{-1}p>p^{r}$. Then, by
(\ref{expansion_Z}) and Definition \ref{Hypothesis_A} we have that
\begin{eqnarray*}
Z(x,t)&=&||x||_{p}^{-n}\left\{(1-p^{-n})\sum_{j=0}^{\infty}e^{-t|\boldsymbol{\psi}_{1}(||x||_{p}^{-1}p^{-j})|^{-\alpha}}p^{-nj}-e^{-t|\boldsymbol{\psi}_{2}(||x||_{p}^{-1}p)|^{-\alpha}}
\right\}\\
&\leq&
||x||_{p}^{-n}\left\{e^{-t|\boldsymbol{\psi}_{1}(||x||_{p}^{-1})|^{-\alpha}}-e^{-t|\boldsymbol{\psi}_{2}(||x||_{p}^{-1}p)|^{-\alpha}}\right\}\\
&\leq& 0.
\end{eqnarray*}
\case $||x||_{p}<p^{-r}$. In this case, we have that
$||x||_{p}^{-1}>p^{r}$. Therefore, proceeding analogously as in the
previous cases, we have that
\begin{eqnarray*}
Z(x,t)&\leq&||x||_{p}^{-n}\left\{e^{-t|\boldsymbol{\psi}_{2}(||x||_{p}^{-1})|^{-\alpha}}-e^{-t|\boldsymbol{\psi}_{2}(||x||_{p}^{-1}p)|^{-\alpha}}\right\}\\
&\leq& 0.
\end{eqnarray*}
\end{proof}

\begin{remark}\label{Application_Z}
Consider the following $p$-adic heat equation or Cauchy problem
given by
\begin{equation}
\left\{
\begin{array}{ll}
\frac{\partial u}{\partial t}(x,t)=\mathcal{A}^{\alpha}u(x,t), \ \  t\in[0,\infty), \ \ x\in\mathbb{Q}_{p}^{n}  \\
&  \\
u(x,0)=u_{0}(x)\in \mathcal{D}(\mathbb{Q}_{p}^{n})\text{, } &
\end{array}
\right.  \label{Cauchy}
\end{equation}
where $\mathcal{A}^{\alpha}$ is the pseudo-differential operator
previously defined.\\
By proceeding as in the proof of \cite[Proposition 1]{To-Z}, we have
that
\begin{equation*}
u(x,t):=\int_{\mathbf{\mathbb{Q}}_{p}^{n}}\chi _{p}\left( -x\cdot
\xi \right)
e^{-t\left[\max\{|\boldsymbol{\psi}_{1}(||\xi||_{p})|,|\boldsymbol{\psi}_{2}(||\xi||_{p})|\}\right]^{-\alpha}}\widehat{u_{0}}(\xi
)d^{n}\xi,
\end{equation*}
with, $u_{0}(x)\in \mathcal{D}(\mathbb{Q}_{p}^{n})$,
$x\in\mathbb{Q}_{p}^{n}$ and $t\geq 0$, is the unique (classical)
solution of the Cauchy problem (\ref{Cauchy}). Moreover, by
\cite[Lemma 3]{To-Z} we have that $u(x,t)=Z_{t}(x)\ast u_{0}$. \\
Since the heat equation describes the distribution of heat (or
temperature variations) in a region along the course of time, then
by Theorem \ref{theorem_Z} we have that equation (\ref{Cauchy})
describes the cooling (or loss of heat) in a given region over time.
\end{remark}


\begin{thebibliography}{99}

\bibitem{Aguilar-Cruz-Estala-2020} Aguilar-Arteaga V., Cruz-L\'opez
M., Estala-Arias S., Non-Archimedean analysis and a wave-type
pseudodifferential equation on finite ad\`{e}les. J. Pseudo-Differ.
Oper. Appl. (2020) DOI: 10.1007/s11868-020-00343-1.

\bibitem{Albeverio et al} Albeverio S., Khrennikov A. Yu., Shelkovich V. M.,
Theory of $p$-adic distributions: linear and nonlinear models.
London Mathematical Society Lecture Note Series, 370. Cambridge
University Press, Cambridge, 2010.

\bibitem{Antoniouk-Khrennikov-Kochubei} Antoniouk A.V., Khrennikov A.Y., Kochubei A.N., Multidimensional nonlinear
pseudo-differential evolution equation with p-adic spatial
variables. J. Pseudo-Differ. Oper. Appl. (2019)
doi:10.1007/s11868-019-00320-3

\bibitem{Berg-Gunnar} Berg Christian, Forst Gunnar, Potential theory on
locally compact abelian groups. Springer-Verlag, New
York-Heidelberg, 1975.

\bibitem{Bikulov-1991} Bikulov A. Kh.,Investigation of the $p$-adic
Green function. Theor. Math. Phys. 87, 376--390 (1991).

\bibitem{C-Ch-G-2020} Casas-S\'{a}nchez, O., Chac\'on-Cort\'{e}s, L., Galeano-Pe\~{n}aloza,
J., Semi-linear Cauchy problem and Markov process associated with a
$p$-adic non-local ultradiffusion operator. J. Pseudo-Differ. Oper.
Appl. 11, 1085–1110 (2020).
https://doi.org/10.1007/s11868-020-00334-2

\bibitem{D-K-K-V} Dragovich B., Khrennikov A. Yu., Kozyrev S. V.,
Volovich I. V., On $p$-adic mathematical physics, $P$-Adic Numbers
Ultrametric Anal. Appl. 1 (1) (2009) 1–17.

\bibitem{Gu-To-1} Guti\'{e}rrez Garc\'{\i}a I., Torresblanca-Badillo A., Strong
Markov processes and negative definite functions associated with
non-Archimedean elliptic pseudo-differential operators. J.
Pseudo-Differ. Oper. Appl. (2019), 1-18.

\bibitem{Gu-To-2} Guti\'{e}rrez Garc\'{\i}a I., Torresblanca-Badillo A., Some classes of non-archimedean pseudo-differential
operators related to Bessel potentials, J. Pseudo-Differ. Oper.
Appl. (2020) DOI: 10.1007/s11868-020-00333-3

\bibitem {K-K-M-2016} Khrennikov A., Oleschko K., Correa L\'opez M., Modeling Fluid's Dynamics with Master Equations in Ultrametric
Spaces Representing the Treelike Structure of Capillary Networks.
Entropy 2016, 18, 249; doi: 10.3390/e18070249

\bibitem{Khrennikov-Kochubei} Khrennikov A. Y., Kochubei A. N., $p-$Adic
Analogue of the Porous Medium Equation. J Fourier Anal Appl. $24$,
1401--1424 (2018).

\bibitem{Khre-Kozy-Zu-libro} Khrennikov A. Yu., Kozyrev S. V., Z\'{u}\~{n}iga-Galindo W. A., Ultrametric pseudodifferential equations
and applications. Encyclopedia of Mathematics and its applications,
Cambridge University Press, 2018. DOI 10.1017/9781316986707

\bibitem{Kochubei-2001} Kochubei A. N., Pseudo-differential equations and
stochastic over non-Archimedean fields, Pure and Applied Mathematics
$244,$ Marcel Dekker, New York,  MR 2003b:35220 Zbl 0984.11063,
2001.

\bibitem{O-K-2017} Oleschko K., Khrennikov A. Yu., "Applications of $p$-adics to
geophysics: Linear and quasilinear diffusion of water-in-oil and
oil-in-water emulsions", TMF, 190:1 (2017), 179–190; Theoret. and
Math. Phys., 190:1 (2017), 154–163

\bibitem{P-K-S-O-C-2019} Pourhadi E., Khrennikov A., Saadati R., Oleschko K., Correa L\'opez M.,
Solvability of the $p$-adic analogue of Navier-Stokes equation via
the wavelet theory. Entropy 2019, 21, 1129; doi: 10.3390/e21111129

\bibitem{Taibleson} Taibleson M. H., Fourier analysis on local fields.
Princeton University Press, 1975.

\bibitem{To-Z-2} Torresblanca-Badillo A., Z\'{u}\~{n}iga-Galindo W. A.,
Non-Archimedean Pseudodifferential Operators and Feller Semigroups,
p-Adic Numbers, Ultrametric Analysis and Applications, Vol. 10, No.
1, pp. 57-73, 2018.

\bibitem{To-Z} Torresblanca-Badillo A., Z\'{u}\~{n}iga-Galindo W. A.,
Ultrametric Diffusion, exponential landscapes, and the first passage
time problem, Acta Appl Math (2018), 157:93.

\bibitem{V-V-Z} Vladimirov V. S., Volovich I. V., Zelenov E. I., $p$-adic
analysis and mathematical physics. World Scientific, 1994.

\bibitem{V-V-1989} Vladimirov V. S., Volovich I. V., $p$-Adic quantum mechanics. Commun. Math. Phys., 123,
659–-676 (1989).

\bibitem{Zu-lib1} Z\'{u}\~{n}iga-Galindo W. A., Pseudodifferential Equations
Over Non-Archimedean Spaces. Lecture Notes in Mathematics 2174,
Springer International Publishing, 2016.





\end{thebibliography}
\end{document}